\documentclass[runningheads]{llncs}
\usepackage{graphicx}
\usepackage{amsmath, amssymb}

\usepackage{optidef}
\usepackage{bm,url}
\usepackage[ruled,vlined]{algorithm2e}

\newcommand{\deq}{\coloneq}

\usepackage{mathtools}
\begin{document}
\title{Fast solution to the fair ranking problem \\ using the Sinkhorn algorithm}
\titlerunning{Fast solution to the fair ranking problem using the Sinkhorn algorithm}
%
\author{Yuki Uehara\inst{1}
\and
Shunnosuke Ikeda\inst{1}
\and
Naoki Nishimura\inst{1}\inst{2}
\and
Koya Ohashi\inst{3}\and \\
Yilin Li\inst{3}\and
Jie Yang\inst{3}\and
Deddy Jobson\inst{3}\and
Xingxia Zha\inst{3}\and
Takeshi Matsumoto\inst{3}\and \\
Noriyoshi Sukegawa\inst{4}
\and
Yuichi Takano\inst{1}}
%
\authorrunning{Y. Uehara et al.}
%
\institute{University of Tsukuba, Tsukuba, Ibaraki 305-8573 Japan\\
\email{\{s2320413,s2330110\}@u.tsukuba.ac.jp, ytakano@sk.tsukuba.ac.jp}\and 
Recruit Co., Ltd, Chiyoda-ku, Tokyo 100-6640 Japan\\
\email{nishimura@r.recruit.co.jp}\and
Mercari, inc, Minato-ku, Tokyo 106-6118 Japan\\
\email{\{k-ohashi,y-li,j-yang,deddy,giselle.T,takeshi.matsumoto\}@mercari.com}\and
Hosei University, Koganei-shi, Tokyo 184-8584 Japan\\
\email{sukegawa@hosei.ac.jp}}
%
\maketitle              
\begin{abstract}
In two-sided marketplaces such as online flea markets, recommender systems for providing consumers with personalized item rankings play a key role in promoting transactions between providers and consumers. 
Meanwhile, two-sided marketplaces face the problem of balancing consumer satisfaction and fairness among items to stimulate activity of item providers.
Saito and Joachims (2022) devised an impact-based fair ranking method for maximizing the Nash social welfare based on fair division; however, this method, which requires solving a large-scale constrained nonlinear optimization problem, is very difficult to apply to practical-scale recommender systems.
We thus propose a fast solution to the impact-based fair ranking problem. 
We first transform the fair ranking problem into an unconstrained optimization problem and then design a gradient ascent method that repeatedly executes the Sinkhorn algorithm. 
Experimental results demonstrate that our algorithm provides fair rankings of high quality and is about 1000 times faster than application of commercial optimization software.

\end{abstract}

\section{Introduction}

A two-sided marketplace is an intermediary economy platform that allows direct interaction between two different user groups. 
As a specific example, we consider an online flea market where providers sell items to consumers. 
In this market, algorithmic recommender systems, which create personalized rankings of attractive items for consumers, play a key role in promoting transactions between providers and consumers.

While most ranking methods have focused particularly on increasing consumer satisfaction, two-sided marketplaces face the problem of balancing consumer satisfaction and fairness among items to stimulate activity of item providers. 
Various definitions, models, and methods have been proposed to ensure fairness in rankings and recommendations; for example, see Pitoura et al.~\cite{pitoura2022fairness} for a systematic survey. 

Most of the prior studies on fair ranking methods have focused on fairness of exposure, which ensures that each item is fairly exposed to consumers~\cite{biega2018equity,singh2018fairness}. 
Patro et al.~\cite{patro2020fairrec} considered fairness of both consumer utility and item exposure in personalized recommendations. 
Togashi et al.~\cite{togashi2024scalable} designed a computationally efficient algorithm to control exposure-based fairness in large-scale recommender systems. 
However, the relationship between item exposure and marketing impacts (e.g., clicks, conversions, and revenues) on items is unclear. 

To ensure fairness of the impacts on recommended items, Saito and Joachims~\cite{saito2022fair} devised an innovative ranking method for maximizing the Nash social welfare (NSW) based on fair division. 
They developed a stochastic ranking policy that guarantees envy-freeness, dominance over uniform rankings, and Pareto optimality. 
However, this policy requires solving a large-scale constrained nonlinear optimization problem. 
Although Saito and Joachims~\cite{saito2022fair} solved the problem using optimization software, it is very difficult to apply this method to practical-scale recommender systems. 

The goal of this paper is to propose a fast solution to the impact-based fair ranking problem. 
For this purpose, we effectively use the Sinkhorn algorithm~\cite{cuturi2013sinkhorn}, which is well known as an efficient solution to the optimal transport problem. 
We first transform the fair ranking problem into an unconstrained optimization problem of computing transport cost matrices that give an optimal solution to the fair ranking problem. 
For this problem, we design a gradient ascent method that repeatedly executes the Sinkhorn algorithm. 

To validate the computational performance of our algorithm, we conducted experiments using synthetic and public real-world datasets. 
Experimental results demonstrate that our algorithm provides fair rankings of high quality and is about 1000 times faster than application of commercial optimization software.

\section{Impact-based fair ranking problem}

We address the impact-based fair ranking problem~\cite{saito2022fair} to provide each consumer with a personalized ranking of items. 
Let $U$ and $I$ be index sets of consumers and items, respectively. 
We consider $\bm{X} \coloneqq (x_{uik})_{(u,i,k) \in U \times I \times [m]} \in \mathbb{R}^{|U| \times|I| \times m}$, where $x_{uik}$ is the probability of item $i \in I$ being ranked at the $k$th position ($k \in [m] \coloneq \{1,2,\ldots,m\}$) for each consumer $u \in U$. 
We formulate a stochastic ranking policy represented by doubly stochastic matrix $\bm{X}_u \coloneqq (x_{uik})_{(i,k) \in I \times [m]} \in \mathbb{R}^{|I| \times m}$ satisfying the following probability constraints: 
\begin{align}
& \sum_{k \in [m]} x_{uik} = 1 \quad (i \in I), \label{eq:dsm1} \\
& \sum_{i \in I} x_{uik} = 1 \quad (k \in [m-1]), \label{eq:dsm2} \\
& x_{uik} \ge 0 \quad (i \in I,~k \in [m]), \label{eq:dsm4}
\end{align}
where the last $m$th position is a dummy (i.e., $\sum_{i \in I} x_{uim} = |I| - m + 1$) to handle a situation where the number of positions is smaller than that of items.  

The impact caused by $\bm{X}_i \coloneqq (x_{uik})_{(u,k) \in U \times [m]} \in \mathbb{R}^{|U| \times m}$ on each item is defined as
\begin{align}
\mathrm{Imp}_i(\bm{X}_i) \deq \sum_{u \in U} \sum_{k \in [m-1]} r(u,i) \cdot e(k) \cdot x_{uik} \quad (i \in I), 
\label{eq:Imp}
\end{align}
where $r(u,i)$ is the relevance of consumer $u$ for item $i$, and $e(k)$ is the exposure probability of the $k$th position. 
The fairness-aware objective function (Nash Social Welfare: NSW) is then defined by the product of the impacts on all items: 
\begin{equation}
F(\bm{X}) \coloneqq \log \prod_{i \in I} \mathrm{Imp}_i(\bm{X}_i) = \sum_{i \in I} \log \mathrm{Imp}_i(\bm{X}_i). 
\label{eq:objF}
\end{equation}
Consequently, the fair ranking problem is formulated as 
\begin{maxi}
  {\bm{X} \in \mathbb{R}^{|U| \times |I| \times m} } {F(\bm{X})}{}{\label{prob:nsw}}
  \addConstraint{}{\mbox{Eqs.~\eqref{eq:dsm1}--\eqref{eq:dsm4}}}{\quad (u \in U).}
\end{maxi}


\section{Gradient ascent method using the Sinkhorn algorithm}

We now describe our algorithm for solving problem~\eqref{prob:nsw}. 
To this aim, we focus on the following optimal transport problem for $u \in U$ with the entropic regularization term~\cite{cuturi2013sinkhorn}: 
\begin{mini}
  {\bm{X}_u \in \mathbb{R}^{|I| \times m}} {\sum_{i \in I} \sum_{k \in [m]} c_{uik} x_{uik} + \varepsilon \sum_{i \in I} \sum_{k \in [m]} x_{uik}(\log x_{uik} - 1)}{}{\label{prob:otp}}
  \addConstraint{}{\mbox{Eqs.~\eqref{eq:dsm1}--\eqref{eq:dsm4}},}
\end{mini}
where $\bm{C}_u \coloneqq (c_{uik})_{(i,k) \in I \times [m]} \in \mathbb{R}^{|I| \times m}$ is a transport cost matrix for $u \in U$, and $\varepsilon > 0$ is a regularization parameter. 
This optimal transport problem can be solved efficiently using the Sinkhorn algorithm~\cite{cuturi2013sinkhorn}. 
Let $\bm{X}^{\star}_u (\bm{C}_u)$ be an optimal solution computed by the Sinkhorn algorithm to problem~\eqref{prob:otp}. 

We now consider searching for transport costs $\bm{C} \coloneqq (\bm{C}_u)_{u \in U}$ such that the corresponding Sinkhorn solution, $\bm{X}^{\star} (\bm{C}) \coloneqq (\bm{X}^{\star}_u (\bm{C}_u))_{u \in U}$, will be optimal for the fair ranking problem~\eqref{prob:nsw}.
Specifically, we transform the constrained optimization problem~\eqref{prob:nsw} into the following unconstrained optimization problem:
\begin{mini}
  {\bm{C} \in \mathbb{R}^{|U| \times |I| \times m}} {F(\bm{X}^{\star} (\bm{C}))}{}{\label{prob:nsw2}}.
\end{mini}

The feasible region of problem~\eqref{prob:nsw2} remains the same in the following sense: 
\begin{theorem} \label{thm:equiv}
    For any $\bm{X}_u$ satisfying Eqs.~\eqref{eq:dsm1}--\eqref{eq:dsm4}, there exists $\bm{C}_u$ such that $\bm{X}_u = \bm{X}_u^{\star}(\bm{C}_u)$. 
\end{theorem}
\begin{proof}
Since the objective function in problem~$(\ref{prob:otp})$ is convex, we
have its optimality condition $c_{uik} + \varepsilon \log x_{uik} = 0$ for $(i,k) \in I \times [m]$.
Hence, for any $\bm{X}_u$, setting $c_{uik} = -\varepsilon \log x_{uik}$ for $(i,k) \in I \times [m]$ makes $\bm{X}_u$ optimal (i.e., $\bm{X}_u = \bm{X}_u^{\star}(\bm{C}_u)$) for problem~$(\ref{prob:otp})$.
\hfill $\square$
\end{proof}

Algorithm~\ref{alg:proposed} describes our gradient ascent method for solving problem~\eqref{prob:nsw2}, where the gradient $\nabla_{\bm{C}} F(\bm{X}^{\star} (\bm{C}))$ can be computed using the chain rule. 
Note also that the Sinkhorn algorithm is compatible with GPU parallel computing. 

\vspace{3mm}
\begin{algorithm}[H]
\caption{Gradient ascent method for solving problem~\eqref{prob:nsw2}}
\label{alg:proposed}

\SetKwInput{KwInput}{Input}
\SetKwInput{KwOutput}{Output}
\SetKwInput{KwInitialize}{Initialize}
\SetKwInput{KwResult}{Result}
\SetKw{KwBreak}{break}
\SetKwFor{ForPar}{for}{do in parallel}{endfor}
\SetKwFor{While}{while}{do}{end}

\SetNlSty{}{}{:}
\SetAlgoNlRelativeSize{0}
\SetAlCapNameFnt{\normalfont\bfseries}
\SetAlCapFnt{\normalfont\bfseries}
\SetAlgoNlRelativeSize{-1}

\KwInput{
  Threshold for optimality \( t > 0 \).
}

\KwInitialize{
    \( \bm{C} \coloneqq (\bm{C}_u)_{u \in U} \in \mathbb{R}^{|U| \times |I| \times m} \) and \( \bm{X} \coloneqq (\bm{X}_u)_{u \in U} \in \mathbb{R}^{|U| \times |I| \times m} \). 
}

\While{\( \|\nabla F(\bm{X})\| > t \)}{
    \ForPar{\( u \in U \)}{
        Compute \( \bm{X}_u = \bm{X}^{\star}_u (\bm{C}_u) \) using the Sinkhorn algorithm. 
    }
    
    
    Update \( \bm{C} \) in the gradient ascent direction \( \nabla_{\bm{C}} F(\bm{X}^{\star} (\bm{C})) \).
}

\KwOutput{
Doubly stochastic matrices \( \bm{X}_u \) for \( u \in U\).
}

\end{algorithm}

\section{Experiments}
In this section, we evaluate the computational performance of our algorithm.

\subsection{Experimental setup}
We used synthetic and public real-world datasets; see  Saito and Joachims~\cite{saito2022fair} for details of these datasets. 
The synthetic dataset was generated following Saito and Joachims~\cite{saito2022fair}. 
The public \texttt{Delicious} dataset was downloaded from the Extreme Classification Repository~\cite{Bhatia16}. 

We use the following evaluation metrics~\cite{saito2022fair}: 
\begin{itemize}
    \item \textbf{User utility}: $\displaystyle \frac{1}{|U|} \sum_{u \in U} \sum_{i\in I} \sum_{k \in [m-1]} r(u,i) \cdot e(k) \cdot X_{uik}$ (larger is better);
    \item  \textbf{Mean max envy}: $\displaystyle \frac{1}{|I|} \sum_{i \in I}  \max_{j \in I} \left(\text{Imp}_i (\bm{X}_j) - \text{Imp}_i (\bm{X}_i) \right)$ (smaller is better);
    \item   \textbf{Items better off (\%)}: Proportion of items for which $\bm{X}$ improves the impact by over 10\% compared to the uniform ranking policy (larger is better);
    \item \textbf{Items worse off (\%)}: Proportion of items for which $\bm{X}$ reduces the impact by over 10\% compared to the uniform ranking policy (smaller is better). 
\end{itemize}

We compare the performance of the following ranking methods\footnote{All these methods were implemented in Python on an Ubuntu 22.04.3 LTS computer equipped with Intel Core i9 12900k CPU 5.2 GHz (128 GB RAM) and NVIDIA GeForce RTX 3090 GPU 1.7 GHz (24 GB RAM). 
Our implementation is available at \url{https://anonymous.4open.science/r/nsw-with-optimal-transport-1C04/}.}: 
\begin{itemize}
\item \textbf{MaxRele}: ranks items in the descending order of relevance;  
\item \textbf{ExpFair(Mosek)}: solves the exposure-based problem~\cite{biega2018equity,singh2018fairness} using Mosek~\cite{mosek};
\item \textbf{NSW(Greedy)}: greedy selects items based on $F(\bm{X})$ for each position; 
\item \textbf{NSW(Mosek)}: solves problem~\eqref{prob:nsw} using Mosek~\cite{mosek};
\item \textbf{NSW(Algo1)}: solves problem~\eqref{prob:nsw2} using our Algorithm~\ref{alg:proposed}; 
\item \textbf{NSW(Algo1+GPU)}: solves problem~\eqref{prob:nsw2} using our Algorithm~\ref{alg:proposed} with GPU, which was used only in this method. 
\end{itemize} 
We implemented the gradient ascent procedure of Algorithm~\ref{alg:proposed} using the PyTorch Adam optimizer~\cite{KingBa15}. 
Note that the results were averaged over five trials. 





\begin{figure}[t]
\centering
\includegraphics[width=0.95\textwidth]{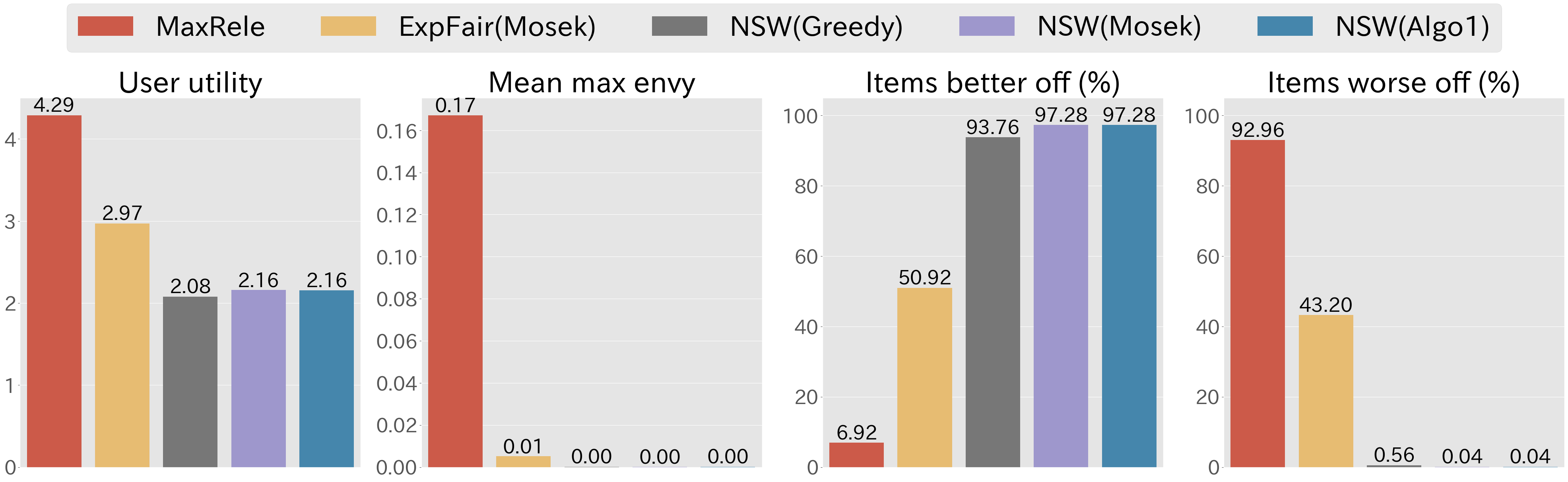}
\caption{Evaluation metrics for the synthetic dataset ($|U|=1000$, $|I|=500$, $m=11$)}
\label{fig:artificial_data}
\end{figure}

\begin{figure}[t]
\centering
\includegraphics[width=0.95\textwidth]{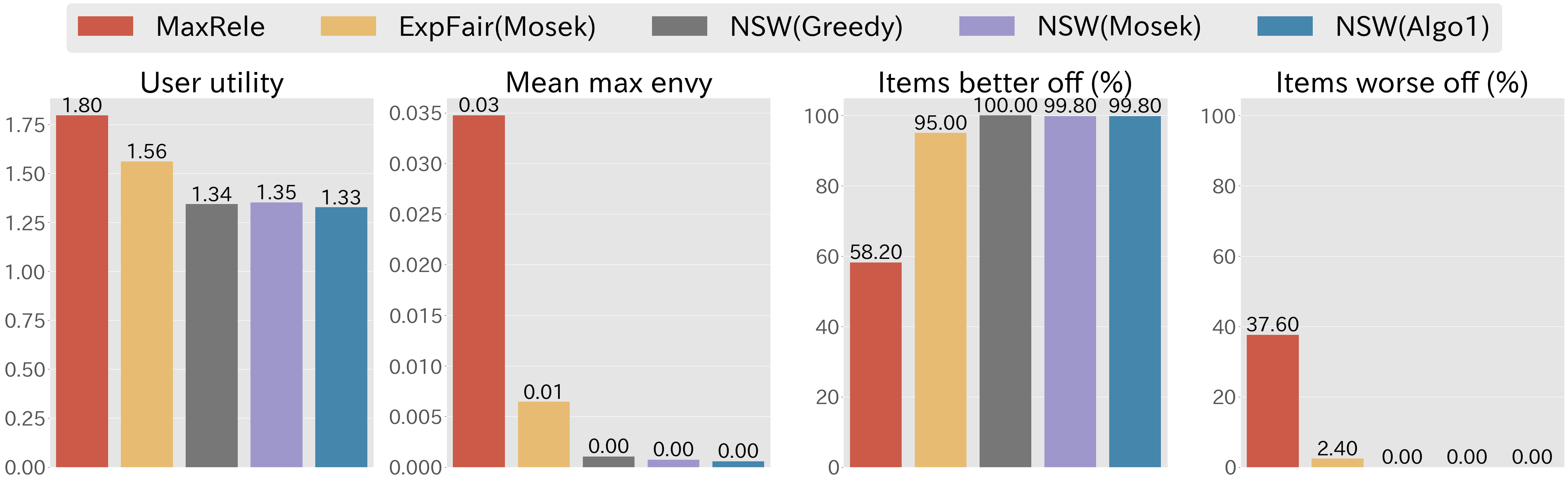}
\caption{Evaluation metrics for the \texttt{Delicious} dataset ($|U|=1014$, $|I|=100$, $m=11$)}
\label{fig:real_world_data}
\end{figure}

\begin{figure}[t]
\centering
\includegraphics[width=0.95\linewidth]{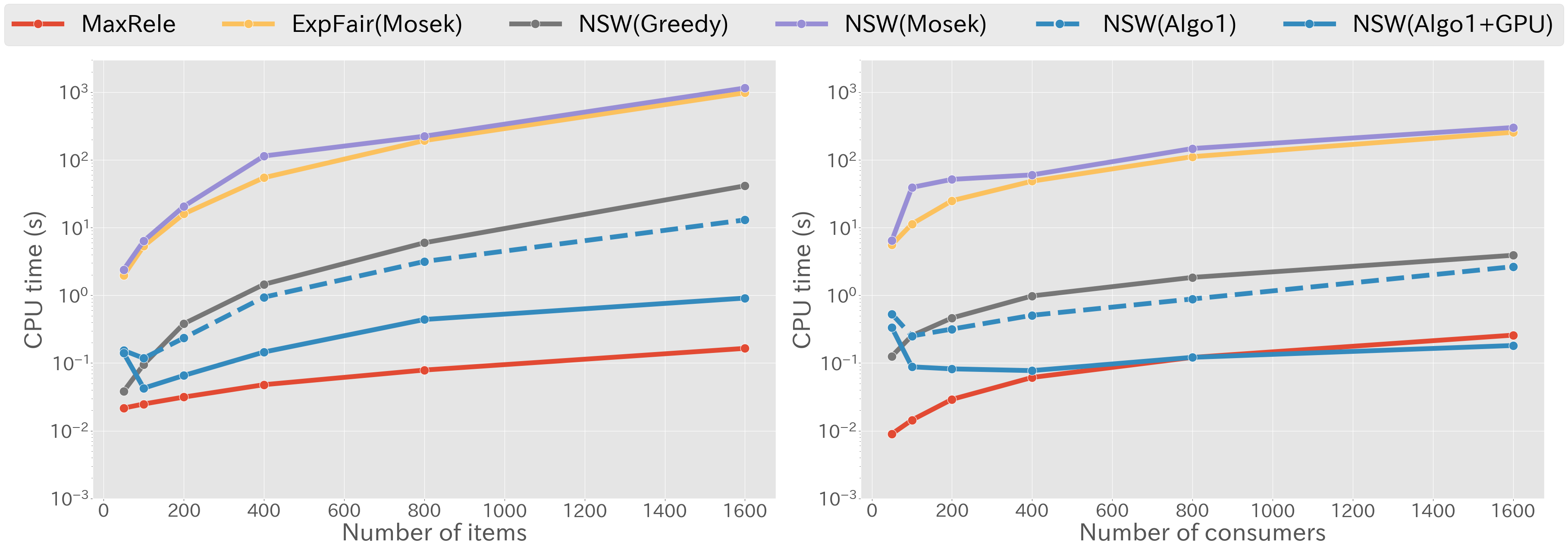}
\caption{Computation times for the synthetic dataset ($|U|=250$, $|I|=250$, $m=11$)}
\label{fig:exec_time}

\end{figure}

\subsection{Results and conclusion}

Figures~\ref{fig:artificial_data} and \ref{fig:real_world_data} show the evaluation results for the synthetic and public \texttt{Delicious} datasets, respectively. 
Our NSW(Algo1) achieved comparable results to NSW(Mosek) and NSW(Greedy). 
In contrast, MaxRele and ExpFair(Mosek) increased user utility, whereas their fairness metrics were very bad. 

Figure~\ref{fig:exec_time} shows the computation times versus the numbers of items and consumers for the synthetic dataset. 
Our NSW(Algo1) was faster than NSW(Greedy) and NSW(Mosek). 
Surprisingly, NSW(Algo1+GPU) was about 10 times faster than NSW(Algo1) and about 1000 times faster than NSW(Mosek). 
Notably, the computation time of NSW(Algo1+GPU) was almost independent of the number of consumers. 

These results demonstrate that our algorithm can provide high-quality solutions significantly faster than other approaches to the impact-based fair ranking problem. 
A future direction of study will be to extend our algorithm to other optimization problems involving doubly stochastic matrices. 



\bibliographystyle{splncs04}
\bibliography{References}
\end{document}